\documentclass[a4paper,onecolumn,10pt,accepted=2025-02-18]{quantumarticle}
\pdfoutput=1
\usepackage[colorlinks=true,bookmarks=false]{hyperref}
\usepackage{amsmath,amsfonts,amssymb,amsthm}
\usepackage{mathtools}
\usepackage{tabularx}
\usepackage{graphicx}
\usepackage[italicdiff]{physics}
\usepackage[sort&compress,numbers]{natbib}
\usepackage{enumitem}
\usepackage{setspace}

\numberwithin{equation}{section}

\newtheorem{theorem}{Theorem}[section]
\newtheorem{corollary}{Corollary}[section]
\newtheorem{proposition}{Proposition}[section]
\newtheorem{lemma}{Lemma}[section]

\newtheorem{definition}{Definition}[section]
\newtheorem{remark}{Remark}[section]

\newcommand{\avg}[1]{\langle#1\rangle}

\newcommand{\Avg}[1]{\left\langle#1\right\rangle}

\newcommand{\bk}[1]{\left(#1\right)}
\newcommand{\Bk}[1]{\left[#1\right]}
\newcommand{\BK}[1]{\left\{#1\right\}}

\newcommand{\mc}[1]{\mathcal #1}
\newcommand{\mb}[1]{\mathbb #1}

\newcommand{\Hc}[1]{\mathrm{H.c.}}

\newcommand{\HS}{\mathrm{HS}}
\newcommand{\Con}{\mathrm{Con}}

\newcommand{\Petz}{\mathrm{Petz}}

\renewcommand{\exp}{\operatorname{exp}}

\newcommand{\umanita}{Umanit{\`a}}
\newcommand{\cgraphic}[2]{\centerline{\includegraphics[width=#1\textwidth]{#2}}}
\newcommand{\fig}[3]{
\begin{figure}[htbp!]
\cgraphic{#1}{#2}
\caption{\label{#2}#3}
\end{figure}
}
\begin{document}

\title{Quantum reversal: a general theory of coherent quantum absorbers}

\author{Mankei Tsang}
\email{mankei@nus.edu.sg}
\homepage{https://blog.nus.edu.sg/mankei/}
\affiliation{Department of Electrical and Computer Engineering,
  National University of Singapore, 4 Engineering Drive 3, Singapore
  117583}
\affiliation{Department of Physics, National University of Singapore,
  2 Science Drive 3, Singapore 117551}
\orcid{0000-0001-7173-1239}

\begin{abstract}
  The fascinating concept of coherent quantum absorber---which can
  absorb any photon emitted by another system while maintaining
  entanglement with that system---has found diverse implications in
  open quantum system theory and quantum metrology. This work
  generalizes the concept by proposing the so-called reversal
  conditions for the two systems, in which a ``reverser'' coherently
  reverses any effect of the other system on a field. The reversal
  conditions are rigorously boiled down to concise formulas involving
  the Petz recovery map and Kraus operators, thereby generalizing as
  well as streamlining the existing treatments of coherent absorbers.
\end{abstract}

\maketitle

\section{\label{sec_intro}Introduction}
A coherent quantum absorber, as conceived by Stannigel \emph{et
  al.}~\cite{stannigel12}, is a system that absorbs any light emitted
by another system while maintaining entanglement with that system. The
idea has since found surprising implications, such as a technique for
finding the steady states of open quantum systems
\cite{stannigel12,roberts20,roberts21} and a method for constructing
efficient continuous measurements for quantum parameter estimation
\cite{godley23,yang23}. Remarkably, Yang, Huelga, and Plenio
\cite{yang23} have recently shown that a method of
measurement-backaction-noise cancellation
\cite{hammerer,qnc,qmfs,khalili18}---which has seen
significant experimental progress with atomic and optomechanical
systems in recent years \cite{moeller17,junker22,jia23}---can be
regarded as a special case of coherent absorbers. This relation
extends the potential impact of the absorber concept to the areas of
magnetometry, optomechanics, and gravitational-wave detectors.

While the absorber concept is fascinating and promising, many special
assumptions and a bewildering array of theoretical tools have been
invoked to study it, such as quantum Markov semigroups
\cite{gardiner_zoller}, cascaded quantum networks \cite{stannigel12},
matrix-product states \cite{yang23}, and quantum detailed balance
\cite{roberts21,fagnola10}. The goal of this work is to
tease out the essential ideas and make the absorber notion more
rigorous as well as generalizable to other scenarios, far beyond the
narrow setting of photon emission and absorption considered in prior
works.

This generality calls for a different name for the generalized
absorber conditions proposed here---I call them the reversal
conditions. Whereas previous works
  \cite{stannigel12,roberts20,roberts21,yang23} assume a
  continuous-time Markov model that requires the interactions between
  the systems and the field to be infinitesimally weak
  \cite{gardiner_zoller}, the reversal conditions here allow the field
  to have arbitrary interaction with the first system before
  interacting with the second. Although Ref.~\cite{godley23} offers a
similar generalization, an appeal of the reversal conditions here is
that they are boiled down to concise but equivalent formulas in
terms of judiciously chosen concepts in open quantum system theory.
Out go the quantum Markov semigroups, the matrix-product states, and
many other extraneous concepts; in come antiunitary operators
\cite{parth_qsc,roberts21}, the Petz recovery map
\cite{petz84,petz,wilde}, and Kraus operators \cite{holevo19} as the
more fundamental ingredients of the reversal conditions.

The involvement of the Petz recovery map here is surprising.  Since
its conception as a quantum generalization of the conditional
expectation \cite{accardi82,petz84}, the Petz map has garnered
increasing attention in diverse areas, including quantum foundations
\cite{leifer13,parzygnat23a}; quantum information theory in the
context of channel sufficiency \cite{petz86,petz}, error correction
\cite{barnum02}, and entropy inequalities \cite{wilde}; quantum
statistical physics in the context of time reversal, detailed balance
\cite{crooks08,fagnola10,duvenhage15,tsang24}, and fluctuation
theorems \cite{manzano15,alhambra18,aberg18,buscemi21}; and even
quantum gravity in the context of entanglement wedge reconstruction
\cite{cotler19,chen20}. The appearance of the Petz map here may well
be a mathematical coincidence, or it may signal a deeper connection
between the reversal concept put forth and the other problems
involving the map.

This work is organized as follows. Sec.~\ref{sec_reversal} introduces
the basic model and defines the reversal condition.  Sec.~\ref{sec_DB}
discusses how a quantum detailed balance condition proposed by Fagnola
and \umanita\ \cite{fagnola10} can simplify the reversal
condition. Sec.~\ref{sec_special} introduces a special reversal
condition that eliminates an ambiguity in the general condition and
may be more useful in experiment design. Theorems~\ref{thm_reversal}
and \ref{thm_reversal2} are the key results that translate the
conditions to the promised formulas. Sec.~\ref{sec_exa} presents two
examples to illustrate
Theorem~\ref{thm_reversal2}---Sec.~\ref{sec_semigroup}, in particular,
discusses a continuous-time Markov model to relate the results here to
the literature on coherent quantum absorbers
\cite{stannigel12,roberts20,roberts21,yang23}.

\section{\label{sec_reversal}Reversal condition}
\subsection{\label{sec_model}Model}
Consider the model depicted by Fig.~\ref{absorber_circuit4}.  Two
systems, denoted as A and B, are initially in a pure state
$\ket{\psi}_{AB}$ in Hilbert space $\mc H_A\otimes\mc H_B$.  A
temporal mode of a traveling field with initial state $\ket{\chi}_E$
in Hilbert space $\mc H_E$ interacts with system A according to a
unitary operator $U$ on $\mc H_A \otimes \mc H_E$.  The field mode
then evolves according to an intermediate unitary operator $U_E$ on
$\mc H_E$, before interacting with system B according to a unitary
operator $V$ on $\mc H_B \otimes \mc H_E$. Note that the system on
$\mc H_E$ is called a temporal field mode only to provide a more
physical picture of the model; it can be an arbitrary system in
reality. All Hilbert spaces are assumed to be finite-dimensional in
the proofs for simplicity.

\fig{0.6}{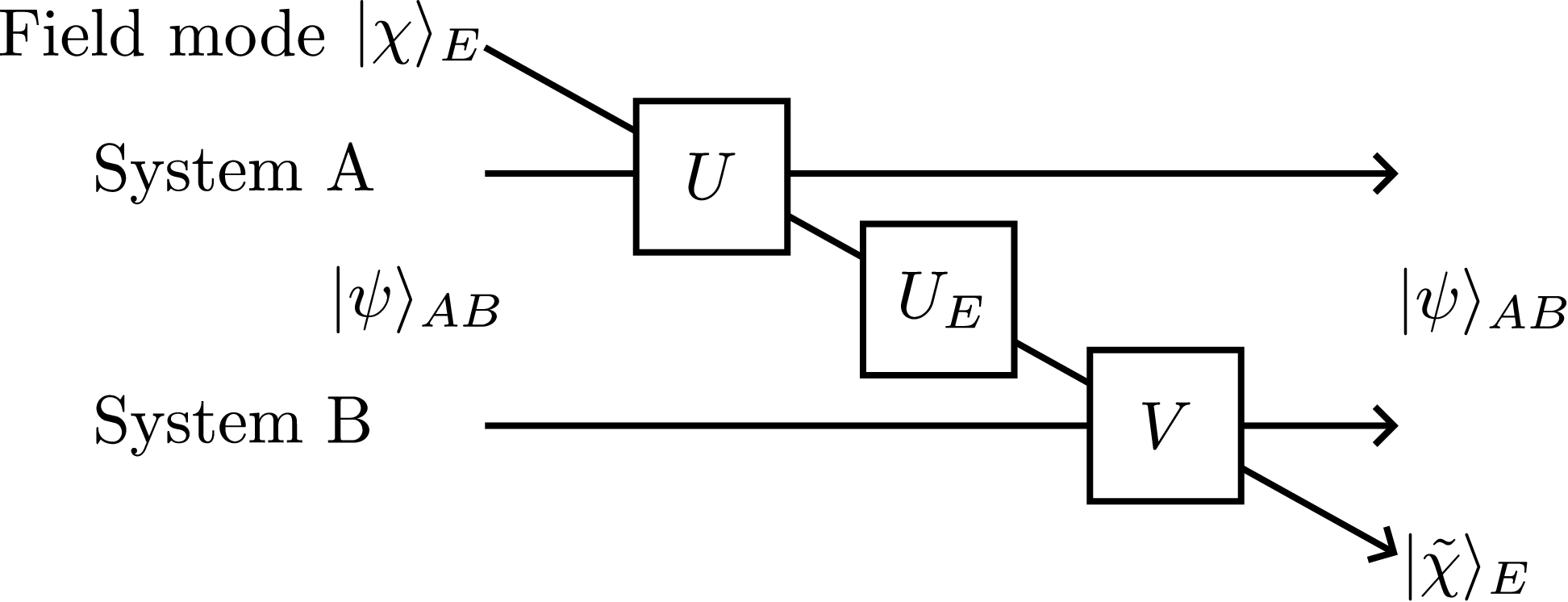}{A quantum circuit illustrating the
  interactions between system A, system B, and the field mode.}

Let $\mc O(\mc H)$ be the set of operators on Hilbert space $\mc H$.
An operator on operators is called a map in this work, also called a
superoperator in the literature. Given $\ket{\chi}_E$ and $U$, a
channel for system A can be modeled by a completely positive,
trace-preserving (CPTP) map $\mc F:\mc O(\mc H_A) \to \mc O(\mc H_A)$
given by
\begin{align}
\Aboxed{\mc F\rho_A &\equiv 
\trace_E\Bk{U (\rho_A\otimes\ket{\chi}_E\bra{\chi}) U^\dagger},}
\label{F}
\end{align}
where $\rho_A$ is an arbitrary density operator on $\mc H_A$,
$\trace_{xy\dots}$ is the partial trace with respect to
$\mc H_x \otimes\mc H_y\otimes\dots$, and subscripts $A$, $B$, and $E$
are used throughout this work to clarify the subspace to which each
expression belongs.  Suppose that a density operator
$\sigma \in \mc O(\mc H_A)$ is a steady state of $\mc F$, viz.,
\begin{align}
\Aboxed{\mc F \sigma &= \sigma.}
\label{steady}
\end{align}
The spectral form of $\sigma$ is assumed to be
\begin{align}
\sigma &= \sum_n p_n\ket{n}_A\bra{n},
\label{sigma}
\end{align}
where $\{p_n:n = 0,\dots,d_A-1\}$ are the eigenvalues of $\sigma$,
$\{\ket{n}_A \in \mc H_A: n = 0,\dots,d_A-1\}$ is an orthonormal
basis of $\mc H_A$, $d_x$ is the dimension of $\mc H_x$, and each
$\ket{n}_A$ is an eigenvector of $\sigma$ with eigenvalue $p_n$.
$\sigma$ is assumed to be full-rank, viz., $p_n > 0$ for all $n$.

Following Roberts~\emph{et al.}~\cite{roberts21}, assume that
$d_B = d_A$ (hereafter abbreviated as $d$) and $\ket{\psi}_{AB}$
is a purification of the steady state $\sigma$, given by
\begin{align}
\ket{\psi}_{AB} &= \sum_n \sqrt{p_n} \ket{n}_A \otimes \ket{\tilde n}_B,
\label{psi}
\end{align}
where 
\begin{align}
\ket{\tilde n}_B &\equiv W \theta \ket{n}_A,
\label{ntilde}
\end{align}
$W:\mc H_A \to \mc H_B$ is a unitary operator, and
$\theta:\mc H_A\to\mc H_A$ is an antiunitary operator. 
The concept of antiunitary operators is reviewed in
Appendix~\ref{app_anti}.

\begin{remark}\label{rem_conj}
  An antiunitary operator can always be decomposed as $u\vartheta$,
  where $u$ is a unitary operator and $\vartheta$ is a conjugation
  (antiunitary and $\vartheta^2 = I_A$, where $I_{xy\dots}$ is the
  identity operator on $\mc H_x\otimes\mc H_y\otimes\dots$).  As $W$
  is left unspecified throughout this work, $u$ may be absorbed into
  the definition of $W$, and $\theta$ in Eq.~(\ref{ntilde}) can be
  taken as a conjugation without loss of generality.
\end{remark}

\subsection{Definition of the reversal condition}

\begin{definition}[Reversal condition]
\label{def_reversal}
Assuming the model in Sec.~\ref{sec_model}, the whole system is said
to obey the reversal condition if there exists a unitary $U_E$ on
$\mc H_E$ such that
\begin{align}
  \Aboxed{(I_A \otimes V)  (I_{AB} \otimes U_E)
(U\otimes I_B) \ket{\psi}_{AB}\otimes\ket{\chi}_E
  &= \ket{\psi}_{AB}\otimes\ket{\tilde\chi}_E,}
\label{V}
\end{align}
where $\ket{\tilde\chi}_E \in \mc H_E$ is the final state of the field
mode.
\end{definition}
Physically, the reversal condition means that system B, as a
``reverser,'' can undo the entanglement between system A and the field
mode, with the help of an intermediate $U_E$. An equivalent picture is
to reverse the arrow of time depicted in Fig.~\ref{absorber_circuit4}
and rewrite Eq.~(\ref{V}) as
\begin{align}
\ket{\psi}_{AB}\otimes\ket{\chi}_E &=
(U^\dagger\otimes I_B)(I_{AB}\otimes U_E^\dagger)(I_A\otimes V^\dagger)
\ket{\psi}_{AB}\otimes\ket{\tilde\chi}_E.
\end{align}
Under this reversed arrow of time,
$\ket{\psi}_{AB}\otimes\ket{\tilde\chi}_E$ is the initial state,
$\ket{\psi}_{AB}\otimes\ket{\chi}_E$ is the final state, and system B
is a reverser \emph{in advance} that interacts with the field mode
first and prevents the entanglement between system A and the field
mode. Either way, system A is always assumed to be given throughout
this work, while system B is to be designed as the reverser.

The allowance of an intermediate $U_E$ is a key property that makes
the reversal condition different from all previously proposed absorber
conditions \cite{stannigel12,roberts20,roberts21,yang23,godley23}, which all
assume $U_E \propto I_E$. The reason for introducing $U_E$ is to make
the reversal condition general enough to be boiled down to a simple
formula in terms of CPTP maps, as shown in Theorem~\ref{thm_reversal}
later.

The original absorber condition proposed by Stannigel \emph{et al.}\
\cite{stannigel12} assumes that both $\ket{\chi}_E$ and
$\ket{\tilde\chi}_E$ are the vacuum state of a bosonic mode. System B
must then absorb any emission by system A---hence the name
absorber. Refs.~\cite{stannigel12,roberts20,roberts21,yang23} also
assume that the system-field interactions are infinitesimally weak so
that a continuous-time Markov model becomes valid; see
Sec.~\ref{sec_semigroup} later for details.  Def.~\ref{def_reversal},
on the other hand, makes no such assumptions and allows the field mode
to have arbitrary dynamics and arbitrary interaction with system A
before interacting with system B. As the condition in
Def.~\ref{def_reversal} is much more general and no longer restricted
to the absorber setting, it is appropriate to give it a different
name.

It can be shown that, given the model in Sec.~\ref{sec_model}, a $V$
that satisfies the reversal condition always exists
\cite[Lemma~4.1]{godley23}; see also Prop.~\ref{prop_godley} later.

After the interactions with the field mode, systems A and B are
assumed, as in previously proposed absorber conditions, to return to
the initial state $\ket{\psi}_{AB}$.  In other words,
$\ket{\psi}_{AB}$ is their steady state after sequential interactions
with multiple field modes, each with the same initial state
$\ket{\chi}_E$. If systems A and B are not initially in this steady
state, they may still converge to it after many rounds of
interactions, such that the assumption of $\ket{\psi}_{AB}$ as the
initial state becomes valid. The precise condition for the
steady-state convergence is, however, outside the scope of this work.

To proceed further, it is vital to define a CPTP map
$\mc G:\mc O(\mc H_B)\to\mc O(\mc H_B)$ as
\begin{align}
\Aboxed{\mc G \rho_B &\equiv 
\trace_E\Bk{V^\dagger (\rho_B\otimes\ket{\tilde\chi}_E\bra{\tilde\chi}) V}.}
\label{G}
\end{align}
Notice the placements of $V^\dagger$ and $V$. Under the arrow of time
depicted in Fig.~\ref{absorber_circuit4}, $\mc G$ is not a physical
channel, but it is a physical channel for system B if the arrow of
time is reversed, so that $V^\dagger$ is the Schr\"odinger-picture
unitary and $\ket{\tilde\chi}_E$ is the ancilla input state. This
definition allows one to express the reversal condition in terms of
the $\mc F$ and $\mc G$ maps, as shown by Theorem~\ref{thm_reversal}
later.

\subsection{A precise formula for the reversal condition}

I now work towards Theorem~\ref{thm_reversal} by presenting a series
of lemmas and definitions.
\begin{lemma}
\label{lem_FG}
The reversal condition is satisfied if and only if the $\mc F$
and $\mc G$ maps defined by Eqs.~(\ref{F}) and (\ref{G}) obey
\begin{align}
(\mc F\otimes \mc I_B)\ket{\psi}_{AB}\bra{\psi} &= 
(\mc I_A\otimes \mc G)\ket{\psi}_{AB}\bra{\psi},
\label{FG}
\end{align}
where $\mc I_{xy\dots}$ is the identity map on
$\mc O(\mc H_x\otimes\mc H_y\otimes\dots)$.
\end{lemma}
\begin{proof}
Define
\begin{align}
\ket{\phi_1} &\equiv (U\otimes I_B) \ket{\psi}_{AB}\otimes\ket{\chi}_E,
\label{pur1}
\\
\ket{\phi_2} &\equiv (I_A \otimes V^\dagger) \ket{\psi}_{AB}\otimes
\ket{\tilde\chi}_E,
\label{pur2}
\end{align}
so that
\begin{align}
(\mc F\otimes \mc I_B)\ket{\psi}_{AB}\bra{\psi}  &= 
\trace_E\bk{\ket{\phi_1}\bra{\phi_1}},
&
(\mc I_A\otimes \mc G)\ket{\psi}_{AB}\bra{\psi} &= 
\trace_E\bk{\ket{\phi_2}\bra{\phi_2}}.
\label{FG2}
\end{align}
To prove the ``only if'' part, note that Eq.~(\ref{V}) implies
\begin{align}
\ket{\phi_2} &= (I_{AB}\otimes U_E) \ket{\phi_1},
\label{phi_UF}
\end{align}
which implies Eq.~(\ref{FG}) via Eqs.~(\ref{FG2}). To prove the ``if''
part, note that Eq.~(\ref{FG}) implies
\begin{align}
\trace_E\bk{\ket{\phi_1}\bra{\phi_1}} &= 
\trace_E\bk{\ket{\phi_2}\bra{\phi_2}}
\label{trace_E_phi}
\end{align}
via Eqs.~(\ref{FG2}), meaning that $\ket{\phi_1}$ and $\ket{\phi_2}$
are purifications of the same density operator on
$\mc H_A\otimes\mc H_B$.  Then there exists a unitary $U_E$ on
$\mc H_E$ such that \cite[Theorem~3.11]{holevo19} 
\begin{align}
\ket{\phi_2} &= (I_{AB}\otimes U_E)\ket{\phi_1},
\end{align}
and Eq.~(\ref{V}) in Def.~\ref{def_reversal} is satisfied.
\end{proof}

To proceed further, it is necessary to establish some linear algebra
first.
\begin{definition}[Hilbert-Schmidt inner product and adjoint]
\label{def_HS}
The Hilbert-Schmidt inner product between two operators $X$ and $Y$
on the same Hilbert space $\mc H$ is defined as
\begin{align}
\avg{Y,X} &\equiv \trace\bk{Y^\dagger X}.
\end{align}
The Hilbert-Schmidt adjoint $\mc M^\HS:\mc O(\mc H') \to \mc O(\mc H)$
of a map $\mc M:\mc O(\mc H) \to \mc O(\mc H')$ is defined by
\begin{align}
\avg{\mc M Y,  X} &= \avg{Y,\mc M^\HS X}
\quad
\forall X,Y.
\end{align}
\end{definition}
\begin{definition}[Connes inner product \cite{connes74} and adjoint]
\label{def_connes}
The Connes inner product between operators $X$ and $Y$ with respect to
a full-rank density operator $\rho$, all on the same Hilbert space
$\mc H$, is defined as
\begin{align}
\avg{Y,X}_\rho &\equiv \avg{Y,\mc E_\rho X} = \trace\bk{Y^\dagger \mc E_\rho X},
\end{align}
where
\begin{align}
\mc E_\rho X &\equiv \rho^{1/2} X \rho^{1/2}
\end{align}
is self-adjoint ($\mc E_\rho = \mc E_\rho^\HS$) and positive-definite
with respect to the Hilbert-Schmidt inner product.

Let $\tau \in \mc O(\mc H')$ be another full-rank density operator.
The adjoint $\mc M^\Con:\mc O(\mc H) \to \mc O(\mc H')$ of a map
$\mc M:\mc O(\mc H') \to \mc O(\mc H)$ with respect to the Connes
inner product is defined by
\begin{align}
\avg{\mc M Y,X}_\rho &= \avg{Y,\mc M^\Con X}_{\tau}
\quad
\forall X \in \mc O(\mc H) ,Y \in \mc O(\mc H').
\label{con}
\end{align}
\end{definition}
Connes introduced this inner product in the context of von Neumann
algebra \cite[Eq.~(1)]{connes74}. The definition here in terms of the
density operator can be found, for example, in
Ref.~\cite[Eq.~(8.17)]{ohya}. It is also called the KMS inner product
in the literature for unknown reasons \cite{carlen17}. The $\Con$
adjoint is instrumental in the works of Accardi and Cecchini
\cite{accardi82} and Petz \cite{petz84,petz}.

\begin{definition}[Petz recovery map \cite{petz84,petz,wilde}]
\label{def_petz}
Given an initial state $\rho \in \mc O(\mc H)$, a CPTP map
$\mc F:\mc O(\mc H)\to \mc O(\mc H')$, and $\tau = \mc F\rho$ for the
$\Con$ adjoint in Def.~\ref{def_connes},
\begin{align}
\mc F^\Petz &\equiv \mc F^{\HS\ \Con\ \HS}
\label{petz}
\end{align}
is called the Petz recovery map. Explicitly,
\begin{align}
\mc F^\Petz &= \mc E_\rho\mc F^\HS  \mc E_{\mc F\rho}^{-1} .
\label{petz_exp}
\end{align}
\end{definition}
Recall Remark~\ref{rem_conj} stating that $\theta$ can always be
assumed to be a conjugation here. It follows that the map defined as
\begin{align}
\Theta X &\equiv \theta^{-1} X \theta = \theta X \theta
\label{Theta}
\end{align}
is also a conjugation with respect to the Hilbert-Schmidt inner
product. Define also the unitary map
$\mc W:\mc O(\mc H_A) \to \mc O(\mc H_B)$ as
\begin{align}
\mc W X &\equiv W X W^\dagger,
\end{align}
and the adjoint map $\mc J:\mc O(\mc H)\to\mc O(\mc H)$ as
\begin{align}
\mc JX &\equiv X^\dagger.
\label{adjoint_map}
\end{align}
These maps will be essential in what follows; their properties are
reviewed in Prop.~\ref{prop_maps}.

Three more lemmas will be needed.
\begin{lemma}
\label{lem_ntilde}
Given Eq.~(\ref{ntilde}),
\begin{align}
\bra{\tilde n}_B X \ket{\tilde m}_B
&= \bra{m}_A \mc J\Theta \mc W^{-1} X \ket{n}_A.
\end{align}
\end{lemma}
\begin{proof}
  To deal with the antiunitary opreator in Eq.~(\ref{ntilde}), I
  switch temporarily from the braket notation to the proper
  inner-product notation $\avg{\cdot,\cdot}$ and write $\ket{n}_A$ as
  $n_A$. The unitarity of $W$ and the conjugation property of
  $\theta$ lead to
\begin{align}
\bra{\tilde n}_B X \ket{\tilde m}_B
&= \avg{W \theta n_A, X W \theta m_A} = 
\avg{W^\dagger X^\dagger W \theta n_A, \theta m_A}
= \avg{m_A,\theta W^\dagger X^\dagger W \theta n_A} 
\\
&= \bra{m}_A \Theta \mc W^{-1} \mc J X \ket{n}_A.
\end{align}
Since $\mc J$ commutes with any unitary or conjugation map by virtue
of Prop.~\ref{prop_maps}, the lemma follows.
\end{proof}
In the following, I always assume a steady state
$\mc F\sigma = \sigma$ for the Connes inner product and adjoint.
The following lemma follows Ref.~\cite{duvenhage15}.
\begin{lemma}[{Ref.~\cite[Eq.~(3)]{duvenhage15}}]
\label{lem_connes}
For any $X,Y\in \mc O(\mc H_A)$,
\begin{align}
\trace\Bk{(X \otimes \mc W\Theta Y) \ket{\psi}_{AB}\bra{\psi}}
&= \avg{Y,X}_\sigma.
\label{lem_connes2}
\end{align}
\end{lemma}
\begin{proof}
\begin{align}
\trace\Bk{(X \otimes \mc W\Theta Y) \ket{\psi}_{AB}\bra{\psi}}
&= 
\sum_{n,m}\sqrt{p_n p_m} 
\bra{n}_A  X\ket{m}_A \bra{\tilde n}_B \mc W \Theta Y\ket{\tilde m}_B
&
(\textrm{by Eq.~(\ref{psi})})
\\
&= 
\sum_{n,m}\sqrt{p_n p_m} 
\bra{n}_A  X\ket{m}_A \bra{m}_A \mc J Y \ket{n}_A
&
(\textrm{by Lemma~\ref{lem_ntilde}})
\\
&= \trace\bk{\sigma^{1/2} X \sigma^{1/2} Y^\dagger}
&
(\textrm{by Eq.~(\ref{sigma})})
\\
&= \avg{Y,X}_\sigma. 
& (\textrm{by Def.~\ref{def_connes}})
\end{align}
\end{proof}

Lemma~\ref{lem_connes} shows the importance of the antiunitary
operator introduced in Eq.~(\ref{ntilde})---the right-hand side of
Eq.~(\ref{lem_connes2}) is antilinear with respect to $Y$, so an
antilinear map $\Theta$ is needed to make the left-hand side
antilinear with respect to $Y$ as well, and the presence of $\Theta$
here can be traced back to Eq.~(\ref{ntilde}) via
Lemma~\ref{lem_ntilde}.

\begin{lemma}[Chain rule]
\label{lem_chain}
\begin{align}
(\mc M_2 \mc M_1)^{*} &= \mc M_1^{*} \mc M_2^{*},
\end{align}
where $*$ is $\HS$, $\Con$, or $\Petz$.
\end{lemma}

All the preceding preparations culminate in the following theorem,
which distills the reversal condition into a precise formula in terms of
the $\mc F$ and $\mc G$ maps, hiding the ``gauge freedom'' of $U_E$ in
the reversal condition.
\begin{theorem}
\label{thm_reversal}
The reversal condition is satisfied if and only if the $\mc F$
and $\mc G$ maps defined by Eqs.~(\ref{F}) and (\ref{G}) obey
\begin{align}
\Aboxed{\mc G &= 
\mc W\Theta \mc F^\Petz \Theta \mc W^{-1},}
\label{G_reversal}
\end{align}
where $\mc F^\Petz$ is the Petz recovery map in Def.~\ref{def_petz}
with respect to the steady state $\sigma = \mc F \sigma$.
\end{theorem}
\begin{proof}
  Consider the left-hand side of Eq.~(\ref{FG}) and write, using
  Lemma~\ref{lem_connes},
\begin{align}
\trace
\Bk{(X\otimes \mc W \Theta Y) (\mc F\otimes \mc I_B)\ket{\psi}_{AB}\bra{\psi}}
&= 
\trace\Bk{(\mc F^\HS X\otimes \mc W \Theta Y) \ket{\psi}_{AB}\bra{\psi}}
= \avg{Y,\mc F^\HS X}_\sigma.
\label{XYF}
\end{align}
The right-hand side of Eq.~(\ref{FG}) similarly gives
\begin{align}
\trace\Bk{(X\otimes \mc W \Theta Y) (\mc I_A\otimes \mc G)
\ket{\psi}_{AB}\bra{\psi}}
&= \avg{ \Theta \mc W^{-1} \mc G^\HS \mc W \Theta Y,X}_\sigma.
\label{XYG}
\end{align}
To prove the ``only if'' part, observe that the reversal condition
implies the equality of Eqs.~(\ref{XYF}) and (\ref{XYG}) for any $X,Y$
by Lemma~\ref{lem_FG}. By the definition of the $\Con$ adjoint in
Def.~\ref{def_connes}, one obtains
\begin{align}
\mc F^{\HS\ \Con} &=  \Theta \mc W^{-1} \mc G^\HS \mc W \Theta.
\end{align}
Taking the Hilbert-Schmidt adjoint, using the chain rule in
Lemma~\ref{lem_chain}, noting that $\mc W^\HS = \mc W^{-1}$ and
$\Theta^\HS = \Theta$, and applying the definition of the Petz map in
Def.~\ref{def_petz}, one obtains
\begin{align}
\mc F^\Petz &= \mc F^{\HS\ \Con\ \HS} =  \Theta \mc W^{-1} \mc G \mc W \Theta,
\end{align}
which leads to Eq.~(\ref{G_reversal}).

To prove the ``if'' part, retrace the preceding steps backwards to go
from Eq.~(\ref{G_reversal}) to
\begin{align}
\trace\Bk{(X\otimes \mc W \Theta Y) (\mc F\otimes \mc I_B)
\ket{\psi}_{AB}\bra{\psi}}
&= 
\trace\Bk{(X\otimes \mc W \Theta Y) (\mc I_A\otimes \mc G)
\ket{\psi}_{AB}\bra{\psi}}
\label{XYFG}
\end{align}
for any $X,Y \in \mc O(\mc H_A)$. Now express an arbitrary
$Z \in \mc O(\mc H_A\otimes\mc H_B)$ as
\begin{align}
Z &= \sum_{j,k} Z_{jk} a_{j} \otimes b_k
\end{align}
in terms of a matrix $Z_{jk}$, a basis $\{a_j\}$ of
$\mc O(\mc H_A)$, and a basis $\{b_k\}$ of
$\mc O(\mc H_B)$.  Plug $X = a_j$ and $Y = \Theta \mc W^{-1}b_k$ into
Eq.~(\ref{XYFG}) and take the sum $\sum_{j,k}Z_{jk}(\dots)$ to obtain
\begin{align}
\trace\Bk{Z (\mc F\otimes \mc I_B)\ket{\psi}_{AB}\bra{\psi}}
&= 
\trace\Bk{Z (\mc I_A\otimes \mc G)\ket{\psi}_{AB}\bra{\psi}}
\quad
\forall Z \in \mc O(\mc H_A\otimes\mc H_B),
\end{align}
which is equivalent to Eq.~(\ref{FG}) and therefore implies the
reversal condition by Lemma~\ref{lem_FG}.
\end{proof}
Eq.~(\ref{G_reversal}) in Theorem~\ref{thm_reversal} is a necessary
and sufficient condition for system B to be a reverser---any reverser
must obey Eq.~(\ref{G_reversal}), and any system B that obeys it is a
reverser.

Before closing this section, I give a noteworthy corollary.
\begin{corollary}
\label{cor_steady2}
Let 
\begin{align}
\tilde\sigma &\equiv \mc W\Theta\sigma = 
\trace_A\bk{\ket{\psi}_{AB}\bra{\psi}}.
\label{sigma2}
\end{align}
Under the reversal condition, $\tilde\sigma$ is a steady
state of $\mc G$, viz.,
\begin{align}
\mc G \tilde\sigma &= \tilde\sigma.
\label{steady2}
\end{align}
\end{corollary}
\begin{proof}
By the recovery property of $\mc F^\Petz$ \cite{petz,wilde},
\begin{align}
\mc F^\Petz \sigma = \mc F^\Petz \mc F \sigma = \sigma.
\label{petz_recovery}
\end{align}
Then use Eqs.~(\ref{G_reversal}), (\ref{sigma2}), and
(\ref{petz_recovery}) to obtain
\begin{align}
\mc G \tilde\sigma &= \mc W \Theta \mc F^\Petz \Theta \mc W^{-1}\mc W \Theta
\sigma = \mc W\Theta \mc F^\Petz \sigma = \mc W\Theta\sigma = \tilde\sigma.
\end{align}
\end{proof}

\section{\label{sec_DB}Detailed balance}
There exist many quantum generalizations of the detailed balance
condition. As discovered by Roberts~\emph{et al.}~\cite{roberts21},
the one most relevant to the absorber theory is the so-called
SQDB-$\theta$ condition proposed by Fagnola and \umanita\
\cite{fagnola10}, where SQDB stands for standard quantum detailed
balance. It turns out that, if system A satisfies the SQDB-$\theta$
condition, the reversal condition given by Theorem~\ref{thm_reversal}
for the whole system can be simplified significantly. Here I define a
discrete-time version of the SQDB-$\theta$ condition for mathematical
simplicity.

\begin{definition}[Discrete-time SQDB-$\theta$ condition]
\label{def_DB}
A CPTP map $\mc F$ and its steady state $\sigma$ are said to satisfy
the SQDB-$\theta$ condition if
\begin{align}
\avg{\mc F^{n\ \HS} Y,X}_\sigma &= \avg{\mc F^{n\ \HS}\Theta X,\Theta Y}_\sigma,
\quad
\forall X,Y \in \mc O(\mc H_A), \quad \forall n = 0,1,2,\dots
\label{DB}
\end{align}
\end{definition}

Physically, $\mc F^n$ is the channel for system A after it interacts
with $n$ temporal field modes sequentially, each with the same initial
state $\ket{\chi}_E$. With further assumptions about $\mc F$, a
continuous-time limit of $\mc F^n$ can be taken to give a quantum
Markov semigroup, which is discussed in more detail in
Sec.~\ref{sec_semigroup}.

The following theorem originates from Ref.~\cite{fagnola10}; I provide
a proof for completeness.
\begin{theorem}[Ref.~\cite{fagnola10}]
\label{thm_DB}
The SQDB-$\theta$ condition in Def.~\ref{def_DB} implies each of the
following two conditions:
\begin{align}
\Aboxed{\sigma &= \Theta\sigma,}
\label{sigma_invar}
\\
\Aboxed{\mc F &= \Theta \mc F^\Petz \Theta.}
\label{DB_petz}
\end{align}
Conversely, the two conditions together imply the SQDB-$\theta$
condition.
\end{theorem}
To prove this theorem, I need the following lemma first.
\begin{lemma}
\label{lem_DB2}
For any $X,Y \in \mc O(\mc H_A)$,
\begin{align}
\avg{\Theta X,\Theta Y}_\sigma &= \avg{Y,X}_{\Theta\sigma},
\label{DB0}
\\
\avg{\mc F^{n\ \HS} \Theta X,\Theta Y}_\sigma
&= \avg{Y,\Theta \mc F^{n\ \HS} \Theta X}_{\Theta\sigma}.
\label{DBn}
\end{align}
\end{lemma}
\begin{proof}
  To prove Eq.~(\ref{DB0}), use Def.~\ref{def_connes} and
  Prop.~\ref{prop_maps} to write
\begin{align}
\avg{\Theta X,\Theta Y}_\sigma &= \avg{\Theta X,\mc E_\sigma \Theta Y}
= \avg{\Theta \mc E_\sigma \Theta Y,X} = \avg{\mc E_{\Theta\sigma} Y,X}
= \avg{Y,X}_{\Theta\sigma}.
\end{align}
To obtain Eq.~(\ref{DBn}), replace $X$ in Eq.~(\ref{DB0}) by
$\Theta \mc F^{n\ \HS} \Theta X$ .
\end{proof}

\begin{proof}[Proof of Theorem~\ref{thm_DB}]
  First prove the forward direction.  To derive
  Eq.~(\ref{sigma_invar}), combine Eq.~(\ref{DB}) for $n = 0$ and
  Eq.~(\ref{DB0}) to obtain
\begin{align}
\avg{Y,X}_\sigma &= \avg{\Theta X,\Theta Y}_\sigma= \avg{Y,X}_{\Theta\sigma}.
\end{align}
Then plug $Y = I$ to obtain
\begin{align}
\avg{I,X}_\sigma &= \trace\bk{\sigma X} = \avg{I,X}_{\Theta\sigma}
= \trace\Bk{(\Theta\sigma) X},
\end{align}
which holds for any $X$, leading to $\Theta\sigma = \sigma$.

For $n > 0$, write
\begin{align}
\avg{\mc F^{n\ \HS}Y,X}_\sigma &= \avg{\mc F^{n\ \HS} \Theta X,\Theta Y}_\sigma
&
(\textrm{by Eq.~(\ref{DB})})
\\
&= \avg{Y,\Theta \mc F^{n\ \HS} \Theta X}_{\sigma},
&
(\textrm{by Eq.~(\ref{DBn}) and $\Theta\sigma = \sigma$})
\end{align}
which implies, by the definition of $\Con$ in Def.~\ref{def_connes},
\begin{align}
\mc F^{n\ \HS\ \Con} &= \Theta \mc F^{n\ \HS} \Theta,
\label{DB1}
\\
\mc F^{n\ \HS\ \Con\ \HS} &= 
\mc F^{n\ \Petz} = \Theta \mc F^{n} \Theta,
\label{DB2}
\end{align}
giving Eq.~(\ref{DB_petz}) for $n = 1$.

Now prove the converse. Eq.~(\ref{sigma_invar}) can be plugged into
Eq.~(\ref{DB0}) to give Eq.~(\ref{DB}) for $n = 0$.  To derive
Eq.~(\ref{DB}) for $n > 0$, take the $n$th power of
Eq.~(\ref{DB_petz}) to obtain Eq.~(\ref{DB2}) by the chain rule in
Lemma~\ref{lem_chain}. Eq.~(\ref{DB2}) is equivalent to
Eq.~(\ref{DB1}), which can be plugged into the right-hand side of
Eq.~(\ref{DBn}).  Plug also $\Theta\sigma = \sigma$ and use the
definition of $\Con$ to obtain
\begin{align}
\avg{\mc F^{n\ \HS} \Theta X,\Theta Y}_\sigma
&= \avg{Y,\mc F^{n\ \HS\ \Con} X}_{\sigma}
= \avg{\mc F^{n\ \HS}Y, X}_{\sigma},
\end{align}
which is Eq.~(\ref{DB}) for $n > 0$.
\end{proof}

The remarkable simplification of the Petz map under the SQDB-$\theta$
condition simplifies the reversal condition as well.

\begin{corollary}
\label{cor_G_DB}
If system A satisfies Eq.~(\ref{DB_petz}), the reversal condition is
satisfied if and only if
\begin{align}
\Aboxed{\mc G &= \mc W \mc F \mc W^{-1}.}
\label{G_DB}
\end{align}
\end{corollary}
\begin{proof}
  Plug Eq.~(\ref{DB_petz}) into Eq.~(\ref{G_reversal}).
\end{proof}

It is interesting to observe that system B mirrors properties of
system A under the reversal condition, as shown by
Corollary~\ref{cor_steady2} and the corollary below.
\begin{corollary}
\label{cor_G_DB2}
  If system A satisfies the SQDB-$\theta$ condition and the whole
  system satisfies the reversal condition, then $\mc G$ and its steady
  state $\tilde\sigma$ also satisfy a SQDB-$\tilde\theta$ condition given by
\begin{align}
\tilde\sigma &= \tilde\Theta\tilde\sigma,
\\
\mc G &= \tilde\Theta \mc G^\Petz \tilde\Theta,
\label{DB_petz2}
\end{align}
where the conjugation $\tilde\theta$ on $\mc H_B$, the conjugation
$\tilde\Theta$ on $\mc O(\mc H_B)$, and $\mc G^\Petz$ are respectively
defined by
\begin{align}
\tilde\theta &\equiv W \theta W^\dagger,
\\
\tilde\Theta X &\equiv \tilde\theta X \tilde\theta = \mc W\Theta\mc W^{-1} X,
\label{Theta2}
\\
\mc G^\Petz &\equiv \mc E_{\tilde\sigma}
\mc G^\HS \mc E_{\mc G\tilde\sigma}^{-1} = \mc E_{\tilde\sigma}
\mc G^\HS \mc E_{\tilde\sigma}^{-1}.
\label{petz2}
\end{align}
\end{corollary}
\begin{proof}
  Use Eqs.~(\ref{Theta2}), (\ref{sigma2}), and (\ref{sigma_invar})
  to write
\begin{align}
\tilde\Theta\tilde\sigma &= \mc W \Theta \mc W^{-1} \mc W \Theta\sigma
= \mc W \sigma = \mc W\Theta\sigma = \tilde\sigma.
\end{align}
One can also show
\begin{align}
\mc G^\Petz &= \mc E_{\tilde\sigma}\mc G^\HS \mc E_{\tilde\sigma}^{-1}
= \mc W\mc E_{\Theta\sigma}\mc W^{-1} \mc G^\HS \mc W\mc E_{\Theta\sigma}^{-1}
\mc W^{-1} 
&
(\textrm{by Eq.~(\ref{sigma2}) and Prop.~\ref{prop_maps}})
\\
&= \mc W\mc E_{\sigma}\mc W^{-1} \mc G^\HS \mc W\mc E_{\sigma}^{-1}
\mc W^{-1} 
&
(\textrm{by Eq.~(\ref{sigma_invar})})
\\
&= \mc W\mc E_{\sigma}\mc F^\HS \mc E_{\sigma}^{-1}\mc W^{-1} 
&
(\textrm{by Eq.~(\ref{G_DB})})
\\
&= \mc W \mc F^\Petz \mc W^{-1} 
&
(\textrm{by Eq.~(\ref{petz_exp}) and $\mc F\sigma = \sigma$})
\\
&= \mc W \Theta \mc W^{-1}\mc G \mc W \Theta  \mc W^{-1}
&
(\textrm{by Eq.~(\ref{G_reversal})})
\\
&= \tilde\Theta \mc G \tilde\Theta,
&
(\textrm{by Eq.~(\ref{Theta2})})
\end{align}
which leads to Eq.~(\ref{DB_petz2}).
\end{proof}

\section{\label{sec_special}Special reversal condition}
A shortcoming of the reversal condition is that it is too vague about
the system-B dynamics $V$ and the intermediate $U_E$ needed to make
system B a reverser. Only certain properties of $V$ are specified
through the relations between the $\mc F$ and $\mc G$ maps, while
$U_E$ is left completely unspecified, making the experimental design
of the reverser difficult.  To overcome the shortcoming, I now focus
on a special reversal condition, where $U_E$ is assumed to be the
identity and the $V$ that makes system B a reverser can be
characterized more explicitly.

\begin{definition}[Special reversal condition]
\label{def_reversal2}
The whole system is said to satisfy the special reversal condition if
Def.~\ref{def_reversal} is satisfied with $U_E = I_E$, viz.,
\begin{align}
  \Aboxed{(I_A \otimes V)(U\otimes I_B) \ket{\psi}_{AB}\otimes\ket{\chi}_E
  &= \ket{\psi}_{AB}\otimes\ket{\tilde\chi}_E.}
\label{V2}
\end{align}
\end{definition}
This definition, apart from allowing $\ket{\tilde\chi}_E$ to be
distinct from $\ket{\chi}_E$, is identical to the absorber condition
assumed by Godley and Guta in their Lemma~4.1 \cite{godley23}. They
have also proved that there always exists a $V$ that satisfies the
condition; I include their result here for completeness.

\begin{proposition}[{Ref.~\cite[Lemma~4.1]{godley23}}]
\label{prop_godley}
A unitary $V$ on $\mc H_B\otimes \mc H_E$ that satisfies
Eq.~(\ref{V2}) always exists.
\end{proposition}
\begin{proof}
  Recall the $\ket{\phi_1}$ defined by Eq.~(\ref{pur1}).  Since
\begin{align}
\trace_{BE}\bk{\ket{\phi_1}\bra{\phi_1}} &= \mc F \sigma = \sigma,
\end{align}
$\ket{\phi_1}$ is a purification of $\sigma$. $\ket{\psi}_{AB}$ is
also a purification of $\sigma$, so the two must be related by an
isometry $R:\mc H_B \to \mc H_B\otimes \mc H_E$ that satisfies
$R^\dagger R = I_B$ as follows \cite{holevo19}:
\begin{align}
\ket{\phi_1} &= I_A \otimes R \ket{\psi}_{AB}.
\end{align}
Now Eq.~(\ref{V2}) is equivalent to
\begin{align}
I_A \otimes V  \ket{\phi_1} &= I_A \otimes V R \ket{\psi}_{AB} 
= \ket{\psi}_{AB} \otimes \ket{\tilde\chi}_E,
\label{VR}
\end{align}
which can be satisfied if 
\begin{align}
V R &= \ket{\tilde\chi}_E,
\label{VR2}
\end{align}
where $\ket{\tilde\chi}_E:\mc H_B\to \mc H_B\otimes \mc H_E$
is a partial ket operator defined by
\begin{align}
\ket{\tilde\chi}_E \ket{\xi}_B &\equiv
\ket{\xi}_B \otimes\ket{\tilde\chi}_E
\quad
\forall \ket{\xi}_B \in \mc H_B.
\end{align}
Eq.~(\ref{VR2}) is equivalent to
\begin{align}
VR \ket{\tilde n}_B &= \ket{\tilde\chi}_E \ket{\tilde n}_B = 
\ket{\tilde n}_B \otimes \ket{\tilde\chi}_E,
\quad
n = 0,\dots,d-1.
\label{VR3}
\end{align}
Observe that $\{R\ket{\tilde n}_B \}$ are an orthonormal set with $d$
elements in $\mc H_B\otimes \mc H_E$, and
$\{\ket{\tilde n}_B \otimes \ket{\tilde\chi}_E \}$ are also an
orthonormal set with $d$ elements in $\mc H_B\otimes \mc H_E$. It
follows that a unitary $V$ that maps the former set to the latter set,
thereby satisfying Eqs.~(\ref{VR3}), (\ref{VR2}), (\ref{VR}), and thus
(\ref{V2}), always exists.
\end{proof}

Since the dimension of $\mc H_B\otimes \mc H_E$ is $d \times d_E$, the
requirement on $V$ given by Eq.~(\ref{VR3}) does not uniquely specify
it. Theorem~\ref{thm_reversal2}, to be shown later in this section, is
a more concrete result, demonstrating that the special reversal
condition can be expressed precisely in terms of Kraus operators.

I now prepare for Theorem~\ref{thm_reversal2} by defining the Kraus
operators needed there and presenting a lemma.

\begin{definition}[Kraus operators]
\label{def_kraus}
Define
\begin{align}
\Aboxed{f_j &\equiv \bra{j}_E U \ket{\chi}_E \in \mc O(\mc H_A),}
\label{f}
\\
\Aboxed{g_j &\equiv \bra{j}_E V^\dagger \ket{\tilde\chi}_E \in \mc O(\mc H_B)}
\label{g}
\end{align}
with respect to the same overcomplete system
$\{\ket{j}_E \in \mc H_E:j=0,\dots,D_E-1\}$ of $\mc H_E$ that
satisfies \cite[Definition~2.20]{holevo19}
\begin{align}
\sum_j \ket{j}_E\bra{j} &= I_E.
\label{IE}
\end{align}
$\{f_j\}$ and $\{g_j\}$ are Kraus operators for the $\mc F$ and
$\mc G$ maps defined by Eqs.~(\ref{F}) and (\ref{G}), respectively,
viz.,
\begin{align}
\mc F \rho_A &= \sum_j f_j \rho_A f_j^\dagger,
\label{f2}
\\
\mc G \rho_B &= \sum_j g_j \rho_B g_j^\dagger.
\label{g2}
\end{align}
\end{definition}
Note that the overcomplete system $\{\ket{j}_E\}$ need not be
orthonormal or even linearly independent---it needs only to satisfy
Eq.~(\ref{IE}).  Note also that
$\ket{\xi}_E:\mc H_x \to \mc H_x \otimes \mc H_E$ and its adjoint
$\bra{\xi}_E:\mc H_x \otimes \mc H_E \to \mc H_x$ in Eqs.~(\ref{f})
and (\ref{g}), where $\mc H_x$ depends on the context, should be
regarded as partial ket and bra operators, as reviewed in
Appendix~\ref{app_braket}.

\begin{lemma}[{Ref.~\cite[Appendix~D]{roberts21}}]
\label{lem_Q}
\begin{align}
X\otimes I_B\ket{\psi}_{AB} &= I_A\otimes Y \ket{\psi}_{AB}
\label{XY_psi}
\end{align}
if and only if
\begin{align}
Y &= \mc W \mc Q X,
\label{YWQX}
\end{align}
where $\mc Q:\mc O(\mc H_A)\to \mc O(\mc H_A)$ is a linear
map defined as
\begin{align}
\Aboxed{\mc Q &\equiv \Theta  \mc J\Delta_\sigma^{-1/2},}
\label{Q}
\end{align}
and the so-called modular map $\Delta_\sigma$ is defined as
\begin{align}
\Delta_\sigma X &\equiv \sigma X\sigma^{-1},
\end{align}
such that
\begin{align}
\Delta_\sigma^{-1/2} X &= \sigma^{-1/2} X\sigma^{1/2}.
\label{mod}
\end{align}
\end{lemma}
\begin{proof}
To prove the ``only if'' part, start from Eq.~(\ref{XY_psi})
and write
\begin{align}
\sum_n \sqrt{p_n} X \ket{n}_A\otimes\ket{\tilde n}_B
&= 
\sum_m \sqrt{p_m} \ket{m}_A\otimes Y\ket{\tilde m}_B,
&
(\textrm{by Eq.~(\ref{psi})})
\\
\sqrt{p_n}\bra{m}_A X\ket{n}_A &= 
\sqrt{p_m} \bra{\tilde n}_B Y \ket{\tilde m}_B
= \sqrt{p_m}\bra{m}_A \mc J\Theta \mc W^{-1} Y \ket{n}_A,
&
(\textrm{by Lemma~\ref{lem_ntilde}})
\\
X \sigma^{1/2} &= \sigma^{1/2} \mc J\Theta \mc W^{-1} Y ,
\end{align}
which leads to Eq.~(\ref{YWQX}). To prove the ``if'' part, it suffices
to show that Eq.~(\ref{YWQX}) leads to
\begin{align}
\bra{n}_A\otimes \bra{\tilde m}_B X\otimes I_B\ket{\psi}_{AB}
&= \bra{n}_A\otimes \bra{\tilde m}_B I_A\otimes Y\ket{\psi}_{AB}
\quad
\forall n,m.
\label{prop_XY_if}
\end{align}
Consider the right-hand side first:
\begin{align}
&\quad 
\bra{n}_A\otimes \bra{\tilde m}_B
I_A\otimes Y\ket{\psi}_{AB}
\nonumber\\
&= 
\bra{n}_A\otimes \bra{\tilde m}_B
I_A\otimes \mc W\Theta\mc J \Delta_\sigma^{-1/2}X\ket{\psi}_{AB}
&
(\textrm{by Eqs.~(\ref{YWQX}) and (\ref{Q})})
\\
&=\bra{n}_A\otimes \bra{\tilde m}_B
I_A\otimes \mc W\Theta\mc J \Delta_\sigma^{-1/2} X
\sum_l\sqrt{p_l} \ket{l}_A\otimes\ket*{\tilde l}_B
&
(\textrm{by Eq.~(\ref{psi})})
\\
&= \sqrt{p_n}\bra{\tilde m}_B 
\mc W\Theta\mc J \Delta_\sigma^{-1/2}X \ket{\tilde n}_B
\\
&= \sqrt{p_n}\bra{n}_A \Delta_\sigma^{-1/2} X \ket{m}_A
&
(\textrm{by Lemma~\ref{lem_ntilde}})
\\
&= \sqrt{p_n}\bra{n}_A \sigma^{-1/2} X \sigma^{1/2} \ket{m}_A
&
(\textrm{by Eq.~(\ref{mod})})
\\
&= \sqrt{p_m} \bra{n}_A X \ket{m}_A.
&
(\textrm{by Eq.~(\ref{sigma})})
\label{XY_step2}
\end{align}
By similar steps, the left-hand side of Eq.~(\ref{prop_XY_if}) is
given by
\begin{align}
\bra{n}_A\otimes \bra{\tilde m}_B X\otimes I_B \ket{\psi}_{AB}
&= \sqrt{p_m} \bra{n}_A X\ket{m}_A,
\label{XY_step3}
\end{align}
which is equal to Eq.~(\ref{XY_step2}), and the desired equality given
by Eq.~(\ref{prop_XY_if}) is proved.
\end{proof}

\begin{theorem}
\label{thm_reversal2}
The special reversal condition is satisfied if and only if the Kraus
operators $\{f_j\}$ and $\{g_j\}$ in Def.~\ref{def_kraus} are related
by
\begin{align}
\Aboxed{g_j &= \mc W \mc Q f_j,}
\label{gf}
\end{align}
where $\mc Q$ is the linear map defined by Eq.~(\ref{Q}).  If
Eq.~(\ref{gf}) holds for Kraus operators defined with respect to one
overcomplete system $\{\ket{j}_E\}$ as per Def.~\ref{def_kraus}, then
it holds for Kraus operators in terms of any overcomplete system of
$\mc H_E$.
\end{theorem}
\begin{proof}
To prove the ``only if'' part, rewrite Eq.~(\ref{V2}) as 
\begin{align}
U \otimes I_B \ket{\psi}_{AB} \otimes \ket{\chi}_E &= 
I_A \otimes V^\dagger   \ket{\psi}_{AB}\otimes \ket{\tilde\chi}_E
\label{V3}
\end{align}
and apply $\bra{j}_E$ on both sides to obtain
\begin{align}
f_j \otimes I_B \ket{\psi}_{AB} &= I_A \otimes g_j \ket{\psi}_{AB},
&
(\textrm{by Def.~\ref{def_kraus}})
\label{fg_step}
\end{align}
which implies Eq.~(\ref{gf}) by Lemma~\ref{lem_Q}. To prove the ``if''
part, note that Eq.~(\ref{gf}) also implies Eq.~(\ref{fg_step}) by
Lemma~\ref{lem_Q}.  Then apply $\sum_j \ket{j}_E$ on both sides of
Eq.~(\ref{fg_step}) to obtain Eq.~(\ref{V3}), which is equivalent to
the special reversal condition.

To prove the invariance of the condition to the assumed overcomplete
system, let $\{f_j'\}$ and $\{g_j'\}$ be the Kraus operators defined
with respect to another overcomplete system
$\{\ket{j'}_E:j = 0,\dots,D_E'-1\}$.  The new system is related to the
old system by
\begin{align}
\bra{j'}_E = \bra{j'}_E \sum_k \ket{k}_E\bra{k} = \sum_k c_{jk} \bra{k}_E,
\end{align}
where $c_{jk} = \braket{j'}{k}$, and the Kraus operators are
similarly related by
\begin{align}
f_j' &= \sum_k c_{jk} f_k,
&
g_j' &= \sum_k c_{jk} g_k.
\end{align}
Now observe that the map $\mc Q$ defined by Eq.~(\ref{Q})
is a linear map, since $\Delta_\sigma^{-1/2}$ 
is linear while $\Theta$ and $\mc J$ are
antilinear.  It follows that $\mc W\mc Q$ is also linear,
and Eq.~(\ref{gf}) implies
\begin{align}
g_j' &= \sum_k c_{jk} g_k
= \sum_k c_{jk} \mc W \mc Q f_k
= \mc W \mc Q\sum_k c_{jk} f_k = \mc W \mc Q f_j',
\end{align}
which is Eq.~(\ref{gf}) in terms of the new Kraus operators.
\end{proof}
Eq.~(\ref{gf}) in Theorem~\ref{thm_reversal2} is a necessary and
sufficient condition for system B to be a special reverser---any
special reverser must obey Eq.~(\ref{gf}), and any system B that obeys
it is a special reverser. 

Notice that $\{\mc Q f_j\}$ are a set of Kraus operators for
$\Theta \mc F^\Petz \Theta$. It follows that $\{\mc W\mc Q f_j\}$ are
a set of Kraus operators for
$\mc W\Theta \mc F^\Petz\Theta\mc W^{-1}$, the right-hand side of
Eq.~(\ref{G_reversal}) in
Theorem~\ref{thm_reversal}. Eq.~(\ref{G_reversal}) in
Theorem~\ref{thm_reversal} is an equality between two maps $\mc G$ and
$\mc W\Theta \mc F^\Petz\Theta\mc W^{-1}$, implying only that their
Kraus operators $\{g_j\}$ and $\{\mc W\mc Q f_j\}$ are related by a
partially isometric matrix \cite{holevo19}.
Theorem~\ref{thm_reversal2}, on the other hand, is a special case of
Theorem~\ref{thm_reversal}, as Eq.~(\ref{gf}) in
Theorem~\ref{thm_reversal2} is an equality between individual Kraus
operators.

The SQDB-$\theta$ condition discussed in Sec.~\ref{sec_DB} can also
simplify the special reversal condition. I need the following lemma
first.
\begin{lemma}
\label{lem_DB_kraus}
If Eq.~(\ref{DB_petz}) is satisfied, there exists a partially
isometric matrix $c \in \mb C^{D_E\times D_E}$ such that
\begin{align}
\mc Q  f_j &= \sum_k c_{jk} f_k.
\label{DB_kraus}
\end{align}
Furthermore, if the set $\{f_j\}$ is linearly independent, then
$\{Q f_j\}$ is also linearly independent and $c$ is
unitary. Conversely, if $\{f_j\}$ satisfies Eq.~(\ref{DB_kraus}) with
a unitary $c$, then Eq.~(\ref{DB_petz}) holds.
\end{lemma}
\begin{proof}
  Eq.~(\ref{DB_petz}) implies the existence of a partially isometric
  $c$ that satisfies Eq.~(\ref{DB_kraus}) by
  Ref.~\cite[Exercise~6.15]{holevo19}.  Now assume that $\{f_j\}$ is
  linearly independent. The linear independence of $\{\mc Q f_j\}$ can
  be proved by contradiction: suppose that $\{\mc Q f_j\}$ is
  linearly dependent. Then there exists a nonzero $a \in \mb C^{D_E}$
  such that $\sum_j a_j \mc Q f_j = 0$. As the $\mc Q$ defined by
  Eq.~(\ref{Q}) is invertible, one can apply $\mc Q^{-1}$ on both
  sides and obtain $\sum_j a_j f_j = 0$, which contradicts the linear
  independence of $\{f_j\}$.  By Ref.~\cite[Exercises~6.14 and
  6.15]{holevo19}, $c$ is unitary if $\{f_j\}$ and $\{\mc Q f_j\}$ are
  both linearly independent.

  The converse part is proved as follows:
\begin{align}
\mc F \rho_A &= \sum_j f_j \rho_A f_j^\dagger
&
(\textrm{by Eq.~(\ref{f2})})
\\
&= \sum_{j,k,l} c_{jk}c_{jl}^* 
(\mc Q f_k) \rho_A  (\mc Q f_l)^\dagger
&
(\textrm{by Eq.~(\ref{DB_kraus})})
\\
&= \sum_k (\mc Q f_k) \rho_A  (\mc Q f_k)^\dagger
&
(\textrm{isometry of $c$})
\\
&= \Theta \mc F^\Petz \Theta \rho_A.
\end{align}
\end{proof}
The matrix $c$ in Lemma~\ref{lem_DB_kraus} needs to be solved on a
case-by-case basis using more details about system A, but once it has
been found, it can be used to simplify Theorem~\ref{thm_reversal2}.
\begin{corollary}
  Suppose that the Kraus operators of the system-A channel
  obey Eq.~(\ref{DB_kraus}) for a certain matrix $c$. Then the special
  reversal condition is satisfied if and only if
\begin{align}
g_j &= \sum_k c_{jk} \mc W f_k.
\end{align}
\end{corollary}
\begin{proof}
Plug Eq.~(\ref{DB_kraus}) into Eq.~(\ref{gf}).
\end{proof}

\section{\label{sec_exa}Examples}
Some examples to illustrate Theorem~\ref{thm_reversal2} are in order.

\subsection{\label{sec_random_unitary}Random unitary channel}
Suppose that the interaction between system A and the field is modeled
by the unitary operator
\begin{align}
U &= \exp[-i (H\otimes I_E + X \otimes Y) t ],
\label{exa_U}
\end{align}
where $t \in \mb R$, $H \in \mc O(\mc H_A)$ is the internal
Hamiltonian of system A, $X \in \mc O(\mc H_A)$ and
$Y \in \mc O(\mc H_E)$ are self-adjoint operators, and
$X\otimes Y$ models the coupling between system A and the
field. Assuming that each $\ket{j}_E$ is an eignevector of $Y$ with
eigenvalue $\lambda_j$, the system-A Kraus operators in
Def.~\ref{def_kraus} become
\begin{align}
f_j &\equiv \bra{j}_E U \ket{\chi}_E = \chi_j u_j,
\quad
j = 0,\dots,d_E - 1,
&
\chi_j &\equiv \braket{j}{\chi}_E,
&
u_j &\equiv \exp[-i (H +\lambda_j X )t].
\end{align}
It follows that the CPTP map $\mc F$ is a random unitary channel given
by
\begin{align}
\mc F \rho_A &= \sum_j f_j \rho_A f_j^\dagger = \sum_j |\chi_j|^2 u_j \rho_A u_j^\dagger.
\end{align}
Assume a steady state of $\mc F$ given by
\begin{align}
\sigma &= \frac{I_A}{d}
\label{exa_sigma}
\end{align}
and an orthonormal basis $\{\ket{n}_A\}$ that obeys
\begin{align}
\theta \ket{n}_A &= \ket{n}_A
\end{align}
to specify the conjugation $\theta$. For system B to be a special
reverser, its Kraus operators as per Def.~\ref{def_kraus} should
satisfy Eq.~(\ref{gf}) in Theorem~\ref{thm_reversal2}, which leads to
\begin{align}
g_j &\equiv \bra{j}_E V^\dagger \ket{\tilde\chi}_E = \mc W \mc Q f_j = 
\mc W \Theta \mc J f_j
= \chi_j \sum_{n,m} \bra{m}_A u_j \ket{n}_A \ket{\tilde n}_B \bra{\tilde m}.
\label{g_exa}
\end{align}
The effect of $\mc W\Theta\mc J$ is to make the $g_j$ matrix with
respect to the $\{\ket{\tilde n}_B\}$ basis, denoted as
$[\bra{\tilde n}_B g_j\ket{\tilde m}_B]$, equal to the transpose of
the $f_j$ matrix $[\bra{n}_A f_j\ket{m}_A]$. The
system-B CPTP map $\mc G$ becomes
\begin{align}
\mc G \rho_B &= \sum_j g_j \rho_B  g_j^\dagger = 
\sum_j |\chi_j|^2 v_j^\dagger \rho_B v_j,
&
v_j &\equiv \mc W \Theta u_j
=
\sum_{n,m} \bk{\bra{n}_A u_j \ket{m}_A}^* \ket{\tilde n}_B \bra{\tilde m},
\end{align}
which is also a random unitary channel. Assuming
\begin{align}
\ket{\tilde\chi}_E &= \exp(-ic t)\ket{\chi}_E,
\quad
c \in \mb R,
\end{align}
one way of implementing the desired system-B Kraus operators is to
make
\begin{align}
V &= e^{-ict} \sum_j v_j \otimes \ket{j}_E\bra{j},
\\
V &= \exp[-i (\tilde H\otimes I_E+ \tilde X \otimes Y) t],
&
\tilde H &= -\mc W \Theta H + c,
&
\tilde X &= -\mc W\Theta X.
\label{exa_V}
\end{align}
The use of a negative-mass oscillator for backaction-noise
cancellation
\cite{hammerer,qnc,qmfs,khalili18,moeller17,junker22,jia23} can then
be seen as a special case when the system-A Hamiltonian satisfies
time-reversal symmetry in the sense of $\Theta H = H$ and
$\tilde H = -\mc W H + c$, meaning that the system-B Hamiltonian
matrix $[\bra{\tilde n}_B\tilde H\ket{\tilde m}_B]$ is the negative of
the $H$ matrix $[\bra{n}_A H \ket{m}_A]$ (ignoring the mathematical
complication of applying the theory here to the infinite-dimensional
systems considered in
Refs.~\cite{hammerer,qnc,qmfs,khalili18,moeller17,junker22,jia23}).

Under the special reversal condition, the steady state of systems A
and B becomes the maximally entangled state
\begin{align}
\ket{\psi}_{AB} &= \frac{1}{\sqrt{d}} \sum_n \ket{n}_A \otimes\ket{\tilde n}_B,
\end{align}
because the system-A steady state given by Eq.~(\ref{exa_sigma}) is
completely mixed. The bases $\{\ket{n}_A\}$ and
$\{\ket{\tilde n}_B = W\theta\ket{n}_A = W \ket{n}_A\}$ are determined
by the chosen $\theta$ and $W$ operators and the system-B dynamics
based on them. For example, let $\{\ket{0}_B,\dots,\ket{d-1}_B\}$
without the tilde be a convenient orthonormal basis of $\mc H_B$
and let $W_1$ and $W_2$ be two examples of $W$ that give
\begin{align}
\ket{\tilde n}_B &= W_1 \theta \ket{n}_A = W_1 \ket{n}_A = \ket{n}_B,
&
\ket{\tilde n}_B &= W_2 \theta\ket{n}_A = W_2 \ket{n}_A = \ket{d-1-n}_B.
\end{align}
Assume further a system-A Hamiltonian given by
\begin{align}
H &= \omega \sum_n n \ket{n}_A \bra{n},
\quad
\omega \in \mb R.
\end{align}
Then the two choices of $W$ lead to two different system-B
Hamiltonians given by
\begin{align}
\tilde H &= -\mc W_1\Theta H + c = c - \omega \sum_n n \ket{n}_B \bra{n},
&
\tilde H &= -\mc W_2 \Theta H + c = c' + \omega \sum_n n \ket{n}_B \bra{n},
\end{align}
where $\mc W_j A \equiv W_j A W_j^\dagger$ and
$c' \equiv c - \omega (d-1)$. In practice, one can choose a $W\theta$
operator either to achieve a desired steady state $\ket{\psi}_{AB}$ or
to make the system-B dynamics easier to implement.

I emphasize that, by virtue of Theorem~\ref{thm_reversal2}, any system
B that implements the Kraus operators given by Eq.~(\ref{gf}) is a
special reverser, so there is considerable freedom in the system
design; Eqs.~(\ref{exa_V}) are simply an obvious choice for this
example. Another possibility is
\begin{align}
V &= \exp[-i (\tilde H\otimes I_E+ \tilde X \otimes Y) \tau],
&
\tilde H &= r\bk{- \mc W \Theta H + c},
&
\tilde X &= -r\mc W\Theta X,
&
r\tau  &= t,
\end{align}
so that the time $\tau \in \mb R$ or the parameter $r \in \mb R$ may
vary as long as $r\tau = t$. More generally, system B need not have a
time-independent Hamiltonian; it may also be a programmable quantum
processor with a time-dependent Hamiltonian to implement
Eq.~(\ref{gf}).

\subsection{\label{sec_semigroup}Continuous-time Markov model}
Suppose that each temporal field mode comprises $J$ spatial modes
labeled by $j = 1,2,\dots,J$, the interaction between system A and
each temporal field mode occurs in infinitesimal time $dt$, and $U$
can be expressed in the form
\begin{align}
U &= I_{AE} + dt G\otimes I_E + \sqrt{dt} 
\sum_{j=1}^J \bk{L_j \otimes a_j^\dagger
- L_j^\dagger \otimes a_j} + o(dt),
\label{U_markov}
\\
G &= -i H -\frac{1}{2} \sum_{j=1}^J L_j^\dagger L_j,
\quad
H = \frac{i}{2} \bk{G-G^\dagger},
\end{align}
where $H \in \mc O(\mc H_A)$ is the internal Hamiltonian of system A,
each $L_j \in \mc O(\mc H_A)$ is the jump operator associated with the
interaction between system A and the $j$th spatial mode,
$a_j \in \mc O(\mc H_E)$ is the annihilation operator for the $j$th
spatial mode, $a_j^\dagger$ is the creation operator, and $o(dt)$
denotes terms that are asymptotically negligible relative to
$dt$. Suppose further that $\ket{0}_E$ denotes the vacuum state for
all $J$ spatial modes, $\ket{j}_E = a_j^\dagger \ket{0}_E$ for
$j= 1,\dots,J$, and $\ket{\chi}_E = \ket{0}_E$.  Then the system-A
Kraus operators in Def.~\ref{def_kraus} can be expressed as
\begin{align}
f_0 &\equiv \bra{0}_E U \ket{\chi}_E = I_A + G dt + o(dt),
\label{f0_markov}
\\
f_j &\equiv \bra{j}_E U \ket{\chi}_E = L_j \sqrt{dt} + o(dt),
\quad
j= 1,\dots,J.
\label{fj_markov}
\end{align}
After sequential interactions with $n$ temporal field modes, the
evolution of system A can be modeled by the CPTP map $\mc F^n$.  A
continuous-time limit with $t = n dt$ can be taken such that $\mc F^n$
becomes a quantum Markov semigroup given by
\begin{align}
\mc F_t &= \exp(\mc L t), \quad t \ge 0,
\end{align}
where $\mc F_t: \mc O(\mc H_A) \to \mc O(\mc H_A)$ for each $t$ is a
CPTP map and $\mc L:\mc O(\mc H_A)\to \mc O(\mc H_A)$ is the semigroup
generator \cite{gardiner_zoller}.  The
Gorini-Kossakowski-Sudarshan-Lindblad (GKSL) form of $\mc L$
\cite{gorini76,lindblad76} becomes
\begin{align}
\mc L \rho &= G \rho + \rho G^\dagger + \sum_{j=1}^J L_j \rho L_j^\dagger.
\end{align}
According to Theorem~\ref{thm_reversal2}, for system B to be a special
reverser, its Kraus operators as per Def.~\ref{def_kraus} must satisfy
\begin{align}
g_0 &\equiv \bra{0}_E V^\dagger \ket{\tilde\chi}_E
= \mc W \mc Q f_0 = I_B + \mc W \mc Q G dt + o(dt),
\\
g_j &\equiv \bra{j}_E V^\dagger \ket{\tilde\chi}_E 
= \mc W \mc Q f_j = \mc W \mc Q L_j \sqrt{dt}+ o(dt),
\quad
j= 1,\dots,J.
\end{align}
Now assume also a continuous-time Markov model for system B, such that
the unitary operator $V$ of system B can be expressed in terms of its
internal Hamiltonian $\tilde H \in \mc O(\mc H_B)$ and jump operators
$\tilde L_1,\dots,\tilde L_J \in \mc O(\mc H_B)$ as
\begin{align}
V &=I_{BE} + dt \tilde G \otimes I_E  + 
\sqrt{dt} \sum_{j=1}^J\bk{\tilde L_j \otimes a_j^\dagger
- \tilde L_j^\dagger \otimes a_j} + o(dt),
\label{V_markov}
\\
\tilde G &= \tilde H - \frac{1}{2}\sum_{j=1}^J \tilde L_j^\dagger \tilde L_j,
\quad
\tilde H = \frac{i}{2} \bk{\tilde G - \tilde G^\dagger}.
\end{align}
Assuming 
\begin{align}
\ket{\tilde\chi}_E = \exp(-i c dt) \ket{0}_E, \quad c \in \mb R,
\end{align}
the desired system-B operators become
\begin{align}
\tilde H &= \frac{i}{2} \bk{\mc J \mc W \mc Q G - \mc W \mc Q G} + c,
&
\tilde L_j &=-\mc W \mc Q L_j.
\label{reversal_GKSL}
\end{align}
Eqs.~(\ref{reversal_GKSL}) are the same as Ref.~\cite[Eqs.~(6) and
(7)]{stannigel12} and Ref.~\cite[Eqs.~(32)]{yang23}.  Note that
Theorem~\ref{thm_reversal2} is much more general, as it allows $U$ to
be an arbitrary unitary, whereas the continuous-time Markov model in
Refs.~\cite{stannigel12,roberts20,roberts21,yang23} requires $U$ to
follow the form of Eq.~(\ref{U_markov}), which imposes stringent
requirements on the field dynamics and the system-field interactions.
For example, Eq.~(\ref{U_markov}) is unable to model cavities for the
field accurately, as the cavity modes often complicate the dynamics.

For the continuous-time model, the SQDB-$\theta$ condition given by
Theorem~\ref{thm_DB} becomes \cite{fagnola10}
\begin{align}
\sigma &= \Theta\sigma,
\label{sigma_invar2}
\\
\mc L &= \Theta \mc L^\Petz \Theta.
\label{DB_gen}
\end{align}
To narrow down the GKSL form for a given semigroup generator $\mc L$,
assume that $\trace(\sigma L_j) = 0$ for all $j$ and
$\{I_A,L_1,\dots,L_J\}$ are linearly independent
\cite{gorini76,parth_qsc}. A GKSL form that satisfies these two
assumptions is called special \cite{fagnola10} (further assumptions
are needed if $\mc H_A$ is infinite-dimensional).  Given
Eq.~(\ref{sigma_invar2}), Eq.~(\ref{DB_gen}) is satisfied if and only
if there exists a special GKSL form of $\mc L$ and a unitary and
involutory matrix $u \in \mb C^{J\times J}$ such that
\cite[Theorem~8]{fagnola10}
\begin{align}
\mc Q G &= G,
&
\mc Q L_j &= \sum_k u_{jk} L_k.
\label{DB_gen2}
\end{align}
These conditions resemble Lemma~\ref{lem_DB_kraus}, although the
precise relation between Eqs.~(\ref{DB_gen2}) and
Lemma~\ref{lem_DB_kraus} is outside the scope of this work.  Given
Eqs.~(\ref{DB_gen2}), the system-B operators in
Eqs.~(\ref{reversal_GKSL}) become
\begin{align}
\tilde H &= -\mc W H + c,
&
\tilde L_j &= -\sum_k u_{jk} \mc W L_k,
\end{align}
which are the same as Ref.~\cite[Eqs.~(54)]{roberts21}.

\section{\label{sec_conclusion}Conclusion}
This work establishes a rigorous theory of quantum reversal by
generalizing the absorber concept and boiling the reversal conditions
down to concise formulas. For future work, it will be important to
work out more examples, explore applications, and study the effect of
decoherence. Of particular interest is the application of the
formalism here to the design of time-reversal-based measurements for
quantum metrology, a problem that has attracted widespread attention
\cite{yang23,godley23,qnc,qmfs,hammerer,khalili18,moeller17,junker22,jia23,yurke86,toscano06,goldstein11,davis16,macri16,li20,colombo22,noise_spec_pra,shi23,gorecki22}.

\section*{Acknowledgments}
I acknowledge inspiring discussions with Dayou Yang, Madalin Guta,
Andrew Lingenfelter, Aash Clerk, Luiz Davidovich, James Gardner,
Yanbei Chen, Rana Adhikari, Arthur Parzygnat, Rocco Duvenhage,
as well as the hospitality of the Kavli Institute for Thereotical
Physics at the University of California-Santa Barbara, the Department
of Physics at the University of Toronto, the Pritzker School of
Molecular Engineering at the University of Chicago, the University of
Warsaw, Caltech, and the Keble College and the Department of Physics
at the University of Oxford, where this work was conceived and
performed. I also thank a referee for suggesting the content in
  Sec.~\ref{sec_semigroup}. This research was supported in part by
the National Research Foundation, Singapore, under its Quantum
Engineering Programme (QEP-P7) and in part by grant NSF PHY-1748958 to
the Kavli Institute for Theoretical Physics (KITP).

\appendix

\section{\label{app_anti}Antiunitary operators and maps}
For clarity, the appendices use the proper notation of an inner
product $\avg{\phi,\psi}$ between two elements $\phi$ and $\psi$ in a
Hilbert space $\mc H$. The inner product is assumed to be antilinear
with respect to the first argument and linear with respect to the
second argument.

An operator $X:\mc H \to \mc H$ is said to be antilinear if
\begin{align}
X(\alpha \psi + \beta \phi)
&= \alpha^* X \psi + \beta^* X \phi
\quad
\forall \psi,\phi \in \mc H, \forall \alpha,\beta \in \mathbb C.
\end{align}
Denote the set of antilinear operators on $\mc H$ as
$\tilde{\mc O}(\mc H)$. The adjoint $X^\dagger$ of an antilinear
operator $X \in \tilde{\mc O}(\mc H)$ is defined by
\begin{align}
\avg{\phi,X\psi} &=\avg{\psi,X^\dagger\phi}
\quad
\forall \phi,\psi \in \mc H.
\end{align}
$X^\dagger$ is also antilinear and $X^{\dagger\dagger} = X$.  The
inverse $X^{-1}$ of an antilinear operator $X$ is defined by
\begin{align}
X^{-1} X = XX^{-1} = I,
\end{align}
where $I$ is the identity operator on $\mc H$.
$X^{-1}$, if it exists, is antilinear.

An operator $T \in \tilde{\mc O}(\mc H)$ is said to be antiunitary if
it obeys
\begin{align}
\avg{T \phi,T\psi} &= \avg{\psi,\phi}
\quad
\forall \phi,\psi.
\end{align}
If $T^2 = I$ in addition, then $T$ is called a conjugation.  I collect
some basic facts about antiunitary operators in the following
proposition; see, for example, Ref.~\cite{parth_qsc}.
\begin{proposition}
\label{prop_antiu}
An antiunitary operator $T$ has the following properties.
\begin{enumerate}
\item $T^\dagger = T^{-1}$, which is antiunitary.
\item Given an orthonormal set $\{e_j\}$ in $\mc H$, $\{T e_j\}$ is another
  orthonormal set.
\item $T$ can always be decomposed as 
\begin{align}
T &= u \theta,
\end{align}
where $u$ is unitary and $\theta$ is a conjugation.  For a
conjugation, there exists a unique orthonormal basis $\{e_j\}$ of
$\mc H$ such that
\begin{align}
\theta e_j &= e_j.
\label{conj_basis}
\end{align}
It follows that
\begin{align}
\theta^\dagger &= \theta^{-1} = \theta,
&
T^\dagger &= T^{-1} = \theta u^\dagger.
\end{align}
\end{enumerate}
\end{proposition}

The inverse $\mc A^{-1}$ and the Hilbert-Schmidt adjoint $\mc A^\HS$
of an antilinear map $\mc A$ are defined in the same manner as those
of an antilinear operator.  I collect some basic facts about unitary
and antiunitary maps in the following proposition.
\begin{proposition}
\label{prop_maps}
Define the $\mc J$, $\Theta$, and $\mc U$ maps on $\mc O(\mc H)$ as
\begin{align}
\mc J X &\equiv X^\dagger,
&
\Theta X &\equiv \theta X \theta,
&
\mc U X &\equiv U X U^\dagger,
\label{def_maps}
\end{align}
where $U$ is an arbitrary unitary operator on $\mc H$ and $\theta$ is
an arbitrary conjugation on $\mc H$.
\begin{enumerate}
\item Commutation with $\mc J$:
\begin{align}
\mc J\Theta &= \Theta \mc J,
&
\mc J \mc U &= \mc U \mc J.
\end{align}
\item Chain rules: for any $X,Y \in \mc O(\mc H)$,
\begin{align}
\mc J(XY) &= (\mc J Y)(\mc JX),
&
\Theta(XY) &= (\Theta X) (\Theta Y),
&
\mc U(XY) &= (\mc U X)(\mc U Y).
\end{align}
\item In terms of the Hilbert-Schmidt inner product, $\mc U$ is
  unitary, while $\mc J$ and $\Theta$ are antiunitary, viz.,
  for any $X,Y \in \mc O(\mc H)$,
\begin{align}
\avg{\mc J X,\mc J Y} &= \avg{Y,X},
&
\avg{\Theta X,\Theta Y} &= \avg{Y,X},
&
\avg{\mc UX,\mc U Y} &= \avg{X,Y}.
\end{align}
It follows that
\begin{align}
\mc J^\HS &= \mc J^{-1} = \mc J,
&
\Theta^\HS &= \Theta^{-1} = \Theta,
&
\mc U^\HS &= \mc U^{-1}.
\end{align}
\item A map in product form
\begin{align}
\mc M_{A,B} X &\equiv A X B
\end{align}
obeys
\begin{align}
\mc M_{A,B}^\HS &= \mc M_{\mc JA,\mc J B},
&
\Theta \mc M_{A,B}   \Theta^{-1} &= \mc M_{\Theta A,\Theta B},
&
\mc U \mc M_{A,B}   \mc U^{-1} &= \mc M_{\mc U A,\mc U B}.
\end{align}
\item Let $\{e_j\}$ be the eigenbasis of $\theta$ satisfying
  Eq.~(\ref{conj_basis}) and let the matrix representation of
  $X \in \mc O(\mc H)$ in this basis be
\begin{align}
X_{jk} &\equiv \Avg{e_j,X e_k}.
\end{align}
Then
\begin{align}
\Avg{e_j,\mc J X e_k} &= X_{kj}^*,
&
\Avg{e_j,\Theta X e_k} &= X_{jk}^*,
&
\Avg{e_j,\Theta \mc J X e_k} &= X_{kj},
\end{align}
where $*$ denotes the entry-wise complex conjugate. 

\item Let $X$ be a normal operator and $s \in \mathbb C$.  Then
\begin{align}
\mc J(X^s) &= (\mc J X)^{s^*},
&
\Theta(X^s)&= (\Theta X)^{s^*},
&
\mc U(X^s) &= (\mc U X)^s.
\end{align}

\end{enumerate}
\end{proposition}

\section{\label{app_braket}Partial ket and bra operators}
Consider two Hilbert spaces $\mc H_x$ and $\mc H_y$ and their tensor
product $\mc H_x\otimes\mc H_y$. I continue to use the proper notation
for the inner product and use the bra and ket notations only to denote
the operators defined in the following. A partial ket operator
$\ket{\xi}_y:\mc H_x\to \mc H_x\otimes \mc H_y$, where
$\xi \in \mc H_y$, is defined by
\begin{align}
\ket{\xi}_y \psi &\equiv \psi \otimes \xi
\quad
\forall \psi \in \mc H_x.
\end{align}
Its adjoint, denoted as $\bra{\xi}_y:\mc H_x\otimes \mc H_y \to \mc H_x$,
is defined by
\begin{align}
\avg{\phi,\ket{\xi}_y \psi}
&= \avg{\bra{\xi}_y \phi, \psi}
\quad
\forall \psi \in \mc H_x, \phi \in \mc H_x\otimes\mc H_y.
\end{align}
I collect some basic facts about the partial ket and bra operators
in the following proposition.
\begin{proposition}
\label{prop_braket}
Let $\xi$ and $\eta$ be arbitrary elements of $\mc H_y$.
\begin{enumerate}
\item  $\bra{\eta}_y \ket{\xi}_y = \avg{\eta,\xi} I_x$.
\item Suppose that $\avg{\xi,\xi} = 1$.  The operator
  $\ket{\xi}_y \bra{\xi}_y:\mc H_x\otimes \mc H_y \to \mc H_x\otimes
  \mc H_y$ is a projection of $\mc H_x\otimes\mc H_y$ onto the
  subspace
\begin{align}
\ket{\xi}_y \mc H_x \equiv \BK{\psi\otimes \xi:\psi \in \mc H_x}.
\end{align}
\item If $\{\xi_j\}$ is an overcomplete system of $\mc H_y$ such that
  $\sum_j \ket{\xi_j}\bra{\xi_j} = I_y$, viz.,
\begin{align}
\psi &= \sum_j \xi_j\avg{\xi_j,\psi}
\quad
\forall \psi \in \mc H_y,
\label{Iy}
\end{align}
then
\begin{align}
\sum_j \ket{\xi_j}_y \bra{\xi_j}_y &= I_x\otimes I_y.
\end{align}
\item For any $X \in \mc O(\mc H_x)$,
\begin{align}
\ket{\xi}_y X &= (X\otimes I_y) \ket{\xi}_y,
&
X \bra{\xi}_y &= \bra{\xi}_y (X\otimes I_y).
\end{align}
\item For any $X \in \mc O(\mc H_x\otimes\mc H_y)$ and any
  overcomplete system $\{\xi_j\}$ of $\mc H_y$,
\begin{align}
\trace_y X = \sum_j \bra{\xi_j}_y X \ket{\xi_j}_y.
\end{align}
\end{enumerate}
\end{proposition}
\begin{proof}
  Let $\psi,\phi$ be arbitrary elements of $\mc H_x$ and let
  $\{a_j\}$, $\{b_k\}$, and $\{a_j\otimes b_k\}$ be orthonormal bases
  of $\mc H_x$, $\mc H_y$, and $\mc H_x\otimes \mc H_y$, respectively.
\begin{enumerate}
\item 
\begin{align}
\avg{\phi,\bra{\eta}_y\ket{\xi}_y\psi} &= \avg{\phi\otimes\eta,\psi\otimes\xi}
= \avg{\phi,\psi}\avg{\eta,\xi}.
\end{align}
\item $\ket{\xi}_y\bra{\xi}_y$ is obviously self-adjoint. It is also
  idempotent because
  $\ket{\xi}_y\bra{\xi}_y \ket{\xi}_y\bra{\xi}_y = \ket{\xi}_y
  \bra{\xi}_y$.  It follows that $\ket{\xi}_y\bra{\xi}_y$ is a
  projection operator.  Now let
  $\mc H_\xi \subseteq \mc H_x\otimes\mc H_y$ be the range of
  $\ket{\xi}_y\bra{\xi}_y$. Let $\Psi$ be an arbitrary element of
  $\mc H_x\otimes \mc H_y$ and observe that
  $\ket{\xi}_y \bra{\xi}_y \Psi \in \ket{\xi}_y \mc H_x$, which
  implies $\mc H_\xi \subseteq \ket{\xi}_y \mc H_x$. Conversely, write
  an arbitrary element of $\ket{\xi}_y \mc H_x$ as $\ket{\xi}_y
  \psi$. Then
\begin{align}
\ket{\xi}_y\bra{\xi}_y \ket{\xi}_y \psi
  &= \ket{\xi}_y\avg{\xi,\xi}\psi = \ket{\xi}_y \psi,
\end{align}
implying that $\ket{\xi}_y \mc H_x\subseteq \mc H_\xi$.  Hence
$\mc H_\xi = \ket{\xi}_y \mc H_x$.

\item Write $\psi = \sum_{k,l} \psi_{kl} (a_k\otimes b_l)$ and
  $\phi = \sum_{m,n} \phi_{mn} (a_m \otimes b_n)$.  Then
\begin{align}
\sum_j \avg{\phi,\ket{\xi_j}_y\bra{\xi_j}_y \psi}
&= \sum_{j,m,n,k,l} \phi_{mn}^*\psi_{kl}
\avg{a_m\otimes b_n,\ket{\xi_j}_y\bra{\xi_j}_y a_k\otimes b_l}
\\
&= \sum_{j,m,n,k,l} \phi_{mn}^*\psi_{kl}\avg{b_n,\xi_j}\avg{\xi_j,b_l} \avg{a_m,a_k}
\\
&= \sum_{m,n,k,l} \phi_{mn}^*\psi_{kl}\avg{b_n,b_l} \avg{a_m,a_k}
\\
&= \sum_{m,n} \phi_{mn}^*\psi_{mn} = \avg{\phi,\psi}.
\end{align}
\item 
\begin{align}
(X\otimes I_y) \ket{\xi}_y  \psi &= (X\otimes I_y) (\psi \otimes \xi)
= (X\psi \otimes \xi) = \ket{\xi}_y X \psi.
\end{align}
\item Let $Y$ be an arbitrary element of  $\mc O(\mc H_x)$.
\begin{align}
\sum_j \trace\bk{ \bra{\xi_j}_y X \ket{\xi_j}_y Y}
&= \sum_{j,k} \avg{a_k,\bra{\xi_j}_y X \ket{\xi_j}_y Y a_k} 
\\
&= \sum_{j,k} \avg{a_k \otimes \xi_j,X (Ya_k \otimes \xi_j)}
\\
&= \sum_{j,k} \avg{a_k \otimes \xi_j,X (Y\otimes I_y) (a_k \otimes \xi_j)}
\\
&= \trace\Bk{X(Y\otimes I_y)} = \trace\Bk{(\trace_y X) Y}.
\end{align}
\end{enumerate}
\end{proof}

\bibliographystyle{quantum}
\bibliography{research3}

\end{document}